\documentclass[preprint,12pt]{elsarticle}

\usepackage{amssymb}
\usepackage{amsmath}
\usepackage{soul}

\usepackage{hyperref}

\hypersetup{
    colorlinks=true,
    linkcolor=blue,
    filecolor=magenta,      
    urlcolor=cyan,
    pdftitle={Overleaf Example},
    pdfpagemode=FullScreen,
    }

\newdefinition{definition}{Definition}[section]
\newdefinition{remark}[definition]{Remark}
\newdefinition{example}[definition]{Example}

\newtheorem{lemma}[definition]{Lemma}
\newtheorem{theorem}[definition]{Theorem}

\newproof{proof}{Proof}

\newcommand{\next}{\operatorname{\mathbf{next}}}

\newcommand{\BB}{\ensuremath{\mathbb{B}}}
\newcommand{\true}{\operatorname{true}}
\newcommand{\false}{\operatorname{false}}

\newcommand{\Tree}{\operatorname{Tree}}
\newcommand{\emptytree}{[]}

\newcommand{\ttt}{^{\scriptscriptstyle t}}
\newcommand{\fff}{^{\scriptscriptstyle f}}
\newcommand{\Xf}{X\fff}
\newcommand{\cons}{\operatorname{::}}
\newcommand{\Xt}{X\ttt}
\newcommand{\phit}{\phi\ttt}
\newcommand{\phif}{\phi\fff}
\newcommand{\epsilont}{\varepsilon\ttt}
\newcommand{\epsilonf}{\varepsilon\fff}

\newcommand{\Type}{\operatorname{Type}}

\newcommand{\xs}{\vec{x}}
\newcommand{\ys}{\vec{y}}

\newcommand{\emptypath}{\operatorname{\langle \rangle}}
\newcommand{\spath}{\operatorname{spath}}
\newcommand{\sigmaf}{\sigma\fff}
\newcommand{\head}{\operatorname{head}}

\newcommand{\s}{s}
\newcommand{\ssf}{\s\fff}

\newcommand{\Path}{\operatorname{\mathrm{Path}}}

\newcommand{\J}{\operatorname{J}}
\newcommand{\K}{\operatorname{K}}
\newcommand{\St}{\operatorname{S}}

\newcommand{\JJ}{\operatorname{\mathcal J}}
\newcommand{\KK}{\operatorname{\mathcal K}}

\newcommand{\Strategy}{\operatorname{Strategy}}
\newcommand{\strategy}{\operatorname{strategy}}
\newcommand{\emptystrategy}{\emptypath}

\newcommand{\Ktimes}{\operatorname{\otimes^{\tiny K}}}
\newcommand{\Jtimes}{\operatorname{\otimes^{\tiny J}}}

\newcommand{\Ksequence}{\operatorname{K{-}sequence}}

\newcommand{\Jsequence}{\operatorname{J{-}sequence}}

\journal{Theoretical Computer Science}

\begin{document}

\begin{frontmatter}



\title{Higher-order Games with Dependent Types}


\author[inst1]{Mart\'\i{n} Escard\'o}

\affiliation[inst1]{organization={School of Computer Science, University of Birmingham},
            addressline={Bristol Road}, 
            city={Birmingham},
            postcode={B15 2TT}, 
            country={UK}}

\author[inst2]{Paulo Oliva}

\affiliation[inst2]{organization={School of Electronic Engineering and Computer Science, Queen Mary University of London},
            addressline={Mile End Road}, 
            city={London},
            postcode={E1 4NS},
            country={UK}}
            
\nopreprintlinetrue

\begin{abstract}
In previous work on higher-order games, we accounted for finite games of unbounded length by working with continuous outcome functions, which carry implicit game trees. In this work we make such trees explicit. We use concepts from dependent type theory to capture history-dependent games, where the set of available moves at a given position in the game depends on the moves played up to that point. In particular, games are modelled by a $W$-type, which is essentially the same type used by Aczel to model constructive Zermelo-Frankel set theory (CZF). We have also implemented all our definitions, constructions, results and proofs in the dependently-typed programming language Agda, which, in particular, allows us to run concrete examples of computations of optimal strategies, that is, strategies in subgame perfect equillibrium.


\end{abstract}

\begin{keyword}
Computational game theory \sep selection function \sep optimal strategies \sep subgame perfect equilibrium \sep dependent type theory \sep Agda.
\end{keyword}

\end{frontmatter}



\paragraph{We dedicate this paper to Ulrich Berger on the occasion of his 65th birthday} Our use of selection functions in this paper, and in our previous work, is inspired by Ulrich's PCF definition of the Fan Functional \cite{Berger(90)}. In particular, the product of selection functions, which we use to compute optimal strategies, is a generalisation of his seminal construction of a selection function for the Cantor space for defining the Fan Functional.

\section{Introduction}
\label{sec:introduction}

Games of perfect information are normally modelled in economics and game theory via \emph{maximising} players. At the end of each play in the game all players receive a payoff, so that players pick their moves in order to maximise their final payoff. In zero-sum games, usually one considers two players and a single outcome, where one player is trying to maximise the outcome, while the other tries to minimise it.

Higher-order games \cite{EO(2010),EO(2011A)} have been introduced as a generalisation of both of these frameworks, where the outcome type of the game can have an arbitrary type $R$, and the goal of each player at a round where $X$ is the set of available moves is defined via generalised \emph{quantifiers}, which are higher-order functions of type $(X \to R) \to R$. 

Standard examples of sequential games of perfect information, such as Tic-Tac-Toe, checkers and chess, can be modelled by taking $R$ to be the linearly ordered set $\{ -1 < 0 < 1 \}$, so that the two players correspond to the quantifiers $\max, \min \colon (X \to R) \to R$ (see \cite{EO(2010D)}). This setup is rather general, in particular allowing to  treat sequential games of unbounded length \cite{EO(2012)}.

In this paper we account for history-dependent sequential games of perfect information of unbounded length, where the set of available moves at a given position in the game depends on the moves played up to that point. This is achieved by modelling game trees as a certain type of well-founded \emph{dependent type trees}. In previous work this was instead done by considering either a fixed game length, or games of countable length with a continuous outcome function \cite{EO(2010)}. In fact, every continuous function carries an implicit well-founded tree, and this goes back to Brouwer. Here we work with explicitly given well-founded trees instead. 

We make use of concepts from \emph{dependent type theory} \cite{Bove-Dybjer(2009)} in order to formalise this, and we have implemented all our definitions, results and proofs in Agda \cite{agda}, which is both a programming language and a proof assistant based on dependent type theory. 

\paragraph{Related work} Our type of trees is essentially the same as that used by Leversha~\cite{Leversha(1976)} to define ordinals, and by Aczel~\cite{Aczel(1978)} to model constructive Zermelo-Frankel set theory (CZF). The precise relationship is discussed in Section~\ref{sec:discussion}.

As in previous work~\cite{EO(2010)}, we think of backward induction in terms of a certain product of selection functions (Definition~\ref{def-sel-fct-prod} below). The product of selection functions is also a manifestation of bar recursion~\cite{EO(2015A)}, used to interpret classical proofs in analysis~\cite{EO(2010)MR2678124, EO(2010)MR2678125}, and can be used to give computational content to the Tychonoff Theorem from topology~\cite{Escardo(2008)}. See our survey paper~\cite{EO(2010D)} where the versatility of the product of selection functions is discussed in further detail. 

Several other recent papers also explore the connection between selection functions and computational interpretations of classical analysis, e.g.\ Powell \cite{Powell(MR3244667), Powell(MR3978430)}.

Ghani, Hedges, Winschel, and Zahn \cite{Ghani(2018)} introduce a notion of \emph{open game}, which is compositional and hence allows to build larger games from smaller blocks, where quantifiers are generalised to best-response functions. 

Hartmann and Gibbons~\cite{hartmann:gibbons:selection} apply selection functions to algorithm design by considering modified notions of (history dependent) products of selection functions, which allows them to design greedy algorithms in a modular way.

Abadi and Plotkin~\cite{abadi:plotkin:selection} consider programming languages for describing systems in terms of choices and their resulting costs and rewards, and use the selection monad to define computationally adequate denotational semantics for them. 

\paragraph{Organisation} Section~\ref{sec:dtt}: Dependent type trees. Section~\ref{sec:games-standard}: Quantifiers and higher-order games with dependent types. Section~\ref{sec:strategies-standard}: Selection functions and their use for the computation of optimal strategies, and main theorem. Section~\ref{sec:discussion}: Discussion of the relationship of our type of type trees with the $W$-type of iterative sets used by Leversha~\cite{Leversha(1976)} to define ordinals in dependent type theory, and by Aczel~\cite{Aczel(1978)}  to model CZF. Section~\ref{sec:formalisation}: Proofs and programs in Agda. 

\section{Dependent type trees}
\label{sec:dtt}

We assume some familiarity with dependent type theory \cite{Bove-Dybjer(2009)}. However, we reason informally, but rigorously, as in the book~\cite{hottbook}. We write the type $\Pi (x : X), A(x)$ as $(x : X) \to A(x)$. The elements of this types are functions that map elements $x : X$ to values in $A(x)$. We write the type $\Sigma (x : X), A(x)$ as $(x : X) \times A (x)$. The elements of this type are (dependent) pairs $(x , y)$ with $x : X$ and $y : A(x)$. 
We work with a universe $\Type$, which is a large type of small types closed under the type theoretic operations. 

We consider sequential game plays of the form $x_0, \dots, x_{n-1}$ where the type of the second move $x_1$ depends on the first move $x_0$, and the type of the third move $x_2$ depends on the first and second moves $x_0$ and $x_1$, and so on. In order to specify such a dependent sequence of moves, we first define dependent-type trees.

\begin{definition}[(Dependent) type trees] We define the type $\Tree$ of \emph{dependent type trees}, or just \emph{type trees} for short, by induction as follows
\begin{enumerate}
    \item There is an empty type tree, denoted by $\emptytree$.
    \item If $X$ is a type and $\Xf \colon X \to \Tree$ is a family of trees indexed by $X$, then there is a type tree with root $X$ and subtrees $\Xf(x)$ for each $x:X$, denoted by $X \cons \Xf$. We think of $\Xf$ as a \emph{forest}.
\end{enumerate}
So notice that $(-) \cons (-)$ is a function of type \[ (X : \Type) \times (X \to \Tree) \to \Tree. \] We adopt the following notational conventions:
\begin{enumerate}
    \item The variable $X$ ranges over types.
    \item The variable $\Xt$ ranges over type trees.
    \item The variable $\Xf$ ranges over type forests.
\end{enumerate}
\end{definition}

\begin{example}[Tic-Tac-Toe type tree] \label{ttt-tree} In Tic-Tac-Toe, the types of available moves at each round form a dependent type tree. The root of the tree is the type of the nine positions on the grid. Depending on the move of the first player, the second player then chooses from the type of the eight remaining available positions, and so on. Once one of the players wins or all the positions are filled then the next player does not have any further available move, arriving at an empty subtree.
\end{example}

A dependent sequence of moves as discussed above is given by a path in a dependent type tree.   

\begin{definition}[Paths]
We define the type $\Path \Xt$ of paths on a type tree $\Xt$ by induction on $\Xt$ as follows:
\begin{enumerate}
    \item There is one path in the empty tree $\emptytree$, which we denote by $\emptypath$.
    \item For every type $X$, every forest $\Xf \colon X \to \Tree$, every $x:X$ and every path $\xs$ in $\Xf(x)$, there is a path $x \cons \xs$ in the tree $X \cons \Xf$. The elements $x$ and $\xs$ are called the \emph{head} and \emph{tail} of the path $x \cons \xs$. 
\end{enumerate}
The above path constructors have the following types:
\begin{enumerate}
    \item $\emptypath : \Path \, \emptytree$.
    \item $(-) \cons (-) : (x : X) \times \Path(\Xf(x)) \to \Path (X \cons \Xf)$.
\end{enumerate}
We let the variable $\xs$ range over paths.
\end{definition}

\begin{example}[Tic-Tac-Toe paths] The paths in the Tic-Tac-Toe tree consists of exactly the sequences of legal moves that end with one of the players winning or a draw.
\end{example}

\newcommand{\structure}{\mathrm{structure}}

We will need to add different structures to the nodes of game trees, including a type of quantifiers (Section~\ref{subsec:quantifiers}), a type of selection functions (Section~\ref{subsec:selection-functions}) and move choices (Section~\ref{sec:strategies-standard}). We accomplish this with a type of structured trees defined as follows, with a parameter $\St$ for the structure, where structure on a type $X$ is taken to be a type $\St X$

\begin{definition}[Trees with structure] \label{def-tree-structure} Given structure $\St : \Type \to \Type$ and a tree $\Xt : \Tree$ we define a new type of structured trees \[ \structure (\St , \Xt) \] by induction on $\Xt$ as follows: 
\begin{enumerate}
    \item There is an empty structured tree \[\emptytree : \structure(\St, \emptytree).\]
    \item If $X$ is a type, $\Xf \colon X \to \Tree$ is a forest and 
    \begin{eqnarray*}
    \s & : & \St X, \\ 
    \ssf & : & (x : X) \to \structure(\St , \Xf(x)), 
    \end{eqnarray*}
    then there is a structured tree \[\s \cons \ssf : \structure(\St , X \cons \Xf)\] with root $X$ equipped with $\s : \St X$ and subtrees $\ssf(x)$ for $x : X$. 
\end{enumerate}
\end{definition}
Notice that in the notation $\s \cons \ssf$ we leave both types $X$ and $\Xf$ implicit, and that the forest $\ssf$ is indexed by the type $X$ rather than the type $\St X$. Although we use the same notation $(-) \cons (-)$, the type of type trees is of course not the same as that of structured type trees, and the latter depends on the former:
\[
\structure : (\Type \to \Type) \times \Tree \to \Type.
\]
Hence if $\Xt$ is a type tree, then $\structure (\St , \Xt)$ is the type of structured type trees over $\Xt$. So for \emph{one} type tree $\Xt : \Tree$ there are, in general, many structured type trees of type $\structure (\St , \Xt)$.

\section{Higher-order games with dependent types}
\label{sec:games-standard}

We start by recalling the notion of a \emph{quantifier}, needed for the definition of a higher-order game (see \cite{EO(2010)}).
For the remainder of the paper we fix a type $R$ of game outcomes. For example, for Tic-Tac-Toe, as discussed in the Introduction, we take $R = \{ -1, 0, 1 \}$, but for other games we can choose~$R$ to be e.g.\ the type of real numbers or the type of booleans.

\subsection{Quantifiers}
\label{subsec:quantifiers}

Let $\K X$ be an abbreviation for the type $(X \to R) \to R$, with $R$ arbitrary but fixed, as described above.

\begin{definition}[Quantifiers] \label{def:quantifier} We call functions of type $\K X$ \emph{quantifiers}.
\end{definition}

If $R$ is the type of booleans, the elements of the function type $\K X$ include the standard existential and universal quantifiers, but we also allow numerical quantifiers, as in the following example.

\begin{example}[$\min$ and $\max$] The simplest non-standard examples of quantifiers are the $\min, \max \colon (X \to R) \to R$ functions, which we will use in our Tic-Tac-Toe running example. The idea is that the type of outcomes $R$ will the linearly ordered type with three values $-1 < 0 < 1$, so that given a type of available moves $X$, the quantifiers select the preferred outcome given a local outcome function $X \to R$. 
\end{example}


\begin{definition}[Quantifier tree] Given a type tree $\Xt$, we define the type of \emph{quantifier trees over $\Xt$} by
\[
\KK \Xt = \structure(\K, \Xt).
\]
We adopt the following notational conventions:
\begin{enumerate}
    \item The variable $\phi$ ranges over the type $\K X$ of quantifiers.
    \item The variable $\phit$ ranges over the type $\KK \Xt$ of quantifier trees.
    \item The variable $\phif$ ranges over quantifier forests $(x \colon X) \to \KK(\Xf(x))$.
\end{enumerate}
\end{definition}


\begin{example}[Tic-Tac-Toe quantifier tree] \label{ttt-q-tree} We associate a quantifier tree to the Tic-Tac-Toe type tree of Example \ref{ttt-tree} by alternating the $\min$ and $\max$ quantifiers depending on the level of the tree. Assuming that player X moves first, and that the outcome $-1$ indicates that player X has won, the root of the quantifier tree is the $\min$ quantifier.
\end{example}

\begin{definition}[Binary dependent product of quantifiers \cite{EO(2010)}] \label{def-quant-prod}
Given a quantifier \[ \phi : \K X \] on a type $X$ and given a family $Y$ of types indexed by~$X$, and a family of quantifiers \[ \gamma \colon (x : X) \to \K(Y(x)),\] we define a quantifier \[ \phi \Ktimes \gamma : \K((x:X) \times Y(x))\] by
\[
(\phi \Ktimes \gamma)(q) = \phi(\lambda x . \gamma(x)(\lambda y . q(x, y))). 
\]
\end{definition}

\begin{example}[Product of $\min$ and $\max$] For instance, for the quantifiers $\min \colon (X \to R) \to R$ and $\max \colon (x : X) \to (Y(x) \to R) \to R$, their product corresponds to a $\min$-$\max$ operation where the range of $y$ depends on $x$:
\[
(\min \Ktimes \max)(q) = \min_{x : X} \max_{y : Y(x)} q(x, y). 
\]
\end{example}

We can now adapt the iterated product of quantifiers \cite{EO(2010)} to quantifier trees.

\begin{definition}[K-sequence] \label{def-K-sequence}
To each quantifier tree \[ \phit : \KK \Xt \] we associate a single quantifier \[\Ksequence\left(\phit\right) : \K (\Path \Xt)\] on the type of paths of the tree $\Xt$ by induction as follows.
\begin{enumerate}
    \item To the empty tree $\emptytree : \Tree$ we associate the quantifier \[ \lambda q . q \emptypath : \K (\Path\emptytree).\]
    \item To a quantifier tree of the form $\phi \cons \phif$ we associate the quantifier \[\phi \Ktimes (\lambda x . \Ksequence\left(\phif(x)\right)).\]
\end{enumerate}
This inductive definition can also be written as
\begin{eqnarray*}
    \Ksequence ([]) & = & \lambda q. \emptypath, \\
    \Ksequence (\phi \cons \phif) & = & \phi \Ktimes (\lambda x . \Ksequence\left(\phif(x)\right)) \\
    & = & \lambda q.\phi(\lambda x . (\Ksequence\left(\phif(x)\right))(\lambda \ys . q(x, \ys))),
\end{eqnarray*}
by definition \ref{def-quant-prod}.
\end{definition}

\subsection{Higher-order games}

Using the notions of type trees and quantifier trees introduced above, we can formulate finite higher-order games as follows.

\begin{definition}[Dependent higher-order game] \label{def:game} A finite \emph{higher-order game} is given by a tuple $(\Xt, R, q, \phit)$ where
\begin{enumerate}
    \item $\Xt$ is a \emph{type tree} of moves available at different stages of the game,
    \item $R$ is the type of possible \emph{outcomes} of the game,
    \item $q \colon \Path(\Xt) \to R$ is the \emph{outcome function},
    \item $\phit$ is a quantifier tree. 
\end{enumerate}
\end{definition}

One can think of the type tree $\Xt$ as describing the rules of the game, specifying the available moves at each position or stage of the game. The function $q \colon \Path(\Xt) \to R$ gives the outcome at the end of the game, for a given sequence of moves. And the quantifier tree $\phit$ describes the objective of the game, assigning a quantifier to each position in the game. The following two examples have the same rules and outcome functions, but different objectives.

\begin{example}[Tic-Tac-Toe] \label{def:ttt-game} The game of Tic-Tac-Toe can be modelled as the higher-order game $(\Xt, R, q, \phit)$ where 
\begin{enumerate}
    \item $\Xt$ is a type tree of moves defined in Example \ref{ttt-tree},
    \item $R$ is the linearly ordered type $-1 < 0 < 1$,
    \item $q \colon \Path(\Xt) \to R$ is the \emph{outcome function} that, given a sequence of moves, has values $-1$ or $1$ if one of the players has won, or $0$ if no more moves are available, and
    \item $\phit$ is a quantifier tree from Example \ref{ttt-q-tree}.  
\end{enumerate}
\end{example}

\begin{example}[Anti-Tic-Tac-Toe] \label{def:anti-ttt-game} The game of Anti-Tic-Tac-Toe, where players try to avoid making three in a row, can be modelled exactly as Tic-Tac-Toe (Example \ref{def:ttt-game}), except that the quantifiers in the quantifier tree $\phit$ from Example \ref{ttt-q-tree} are swapped. The quantifiers still alternate between $\min$ and $\max$, but now the root of the quantifier tree is the $\max$ quantifier.
\end{example}

Consider also the following example, where the quantifier tree is not based on $\min$ and $\max$, but instead on existential quantifiers.

\begin{example}[Eight Queens] \label{def:eight-queens-game} The game of Eight Queens, where a player attempts to place eight queens on a chess board in such way that no queen can capture the other ones, can be modelled as the higher-order game $(\Xt, R, q, \phit)$ as follows:
\begin{enumerate}
    \item $\Xt$ is a type tree of moves with depth 8, describing the positions on the chess board that have not been used yet.
    \item $R$ is the type of Booleans $\BB = \{ \true, \false \}$.
    \item $q \colon \Path(\Xt) \to R$ is the \emph{outcome function} that, given the position of the eight queens on the board, has values $\false$ or $\true$ depending on whether some queen can capture another or not.
    \item $\phit$ is the quantifier tree of existential quantifiers $\exists \colon (X \to \BB) \to \BB$ at each node, where $X$ is the type of available moves at the node, defined by $\exists(p)=\true$ if and only if there is $x:X$ such that $p(x)=\true$.
\end{enumerate}
\end{example}

Notice that, in the above example, in some sense there is only one player, although the concept of player is not incorporated explicitly in our notion of game.

\begin{definition}[Optimal outcome] Given a game $(\Xt, R, q, \phit)$ we call $\Ksequence\left(\phit\right)(q) : R$ the \emph{optimal outcome} of the game.
\end{definition}

\begin{example} In the case of Tic-Tac-Toe, it is well-known that the optimal outcome of the game is a draw. Perhaps less well-known is that in the case of Anti-Tic-Tac-Toc the optimal outcome is also a draw. It is too laborious to calculate these by hand using our definitions, but the definitions can be directly translated to e.g.\ Haskell or Agda programs so that optimal outcomes of games can be automatically calculated~\cite{EO(2010D),TypeTopologyGames}.
\end{example}

\section{Optimal strategies}
\label{sec:strategies-standard}

A strategy for a given type tree selects an element of each node of the tree. Using Definition~\ref{def-tree-structure}, this corresponds to the identity structure $S X = X$. That is, for each type $X$ of moves in the game tree, the structure is an element of $S X$, which amounts to saying that we pick an element of~$X$.
\begin{definition}[Type of strategies] 
Given a type tree $\Xt$, we define the type of \emph{strategies over $\Xt$} by
\[
\Strategy \Xt = \structure((\lambda X. X), \Xt).
\]
\end{definition}

The type $\Strategy\emptytree$ has exactly one element $\emptystrategy$, which we think of as the empty strategy. Given a type tree of the form $X \cons \Xf$, an element of the type of strategies 
\[ \Strategy (X \cons \Xf) =  X \times ((x : X) \to \Strategy (\Xf(x))), \]
is a pair consisting of a move in $X$ together with a function that maps each such move $x : X$ to a strategy in the subtree $\Xf(x)$. 

\begin{definition}[Strategic path] Each $\sigma \colon \Strategy \Xt$ determines a single path \[ \spath(\sigma) : \Path \Xt,\] which we call the \emph{strategic path of $\sigma$}, as follows:
\begin{enumerate}
    \item If $\sigma : \Strategy \, \emptytree$ then its strategic path is the empty path~$\emptypath$.
    \item If $\sigma : \Strategy (X \cons \Xf)$ then $\sigma$ is of the form $x \cons \sigmaf$ and its strategic path is defined as $x \cons \spath (\sigmaf(x))$.  
\end{enumerate}
\end{definition}

\begin{definition}[Optimal strategies] \label{def:strategy} Given a game $(\Xt, R, q, \phit)$, we define the notion of optimal strategy by induction on $\Xt$ as follows.
\begin{enumerate}
    \item If $\Xt$ is the empty tree $\emptytree$ then any $\sigma :\strategy(\Xt)$ is empty, which is considered optimal,
    \item When the game has a type tree of the form $X \cons \Xf$ with a quantifier tree for the form $\phi \cons \phif$, then a strategy $x_0 \cons \sigmaf$ is optimal if the following two conditions hold, where
    $q_x  =  \lambda \ys . q(x \cons \ys)$:
    \vspace{1ex}
    \begin{enumerate}
        \item $q_{x_0}(\spath (\sigmaf(x_0))) = \phi(\lambda x . q_x(\spath (\sigmaf(x))))$.
        \vspace{1ex}
    \item For each $x : X$, the strategy $\sigmaf(x) \colon \strategy(\Xf(x))$ is optimal at the subgame $(\Xf(x), R, q_x, \phif(x))$.
    \end{enumerate}
\end{enumerate}
This is a generalisation of the notion of subgame perfect equilibrium, where quantifiers other than $\min$ and $\max$ are allowed.
\end{definition}
    
\begin{theorem}[Optimality theorem] Given a game $(\Xt, R, q, \phit)$ and an optimal strategy $\sigma$ for it, we have that the outcome on the strategic path of $\sigma$ is indeed the optimal outcome of the game, that is,
\begin{equation*}
    q(\spath \sigma) = \Ksequence(\phit)(q).
\end{equation*}
\end{theorem}

\begin{proof} By a induction on the type tree $\Xt$. The base case when $\Xt$ is the empty tree is immediate. When the game is defined over a type tree $X \cons \Xf$, let $\phi \cons \phif$ be the associated quantifier tree, and $x_0 \cons \sigmaf$ the given strategy. Then, by the induction hypothesis, we have that for every $x : X$ 
\[ q_x(\spath (\sigmaf(x))) = \Ksequence (\phif(x))(q_x) \]
and hence
\[
\begin{array}{rcl}
    q(\spath (x_0 \cons \sigmaf)) 
        & = & q_{x_0}(\spath (\sigmaf(x_0))) \\[2mm]
        & = &  \phi(\lambda x . q_x(\spath (\sigmaf(x)))) \\[1mm]        
        & \stackrel{({\rm IH})}{=} & \phi(\lambda x . (\Ksequence\left(\phif(x)\right))(q_x)) \\[2mm]
        & = & \Ksequence (\phi \cons \phif) (q)
\end{array}
\]
using in the first step that $x_0 \cons \sigmaf$ is an optimal strategy. \hfill $\Box$
\end{proof}


\subsection{Selection functions}
\label{subsec:selection-functions}

Let $\J X$ be an abbreviation for the type $(X \to R) \to X$. Recall the following definitions from~\cite{EO(2010)}.

\begin{definition}[Selection functions] \label{def:selection} \leavevmode
\begin{enumerate}
    \item We call functions of type $\J X$ \emph{selection functions}.
    \item We say that a selection function $\varepsilon : \J X$ attains a quantifier $\phi : \K X$ if for all $p \colon X \to R$.
    \begin{equation*}
        p(\varepsilon(p)) = \phi(p).
    \end{equation*}
    \item For each selection function $\varepsilon : \J X$ we can associate a quantifier $\overline{\varepsilon} : \K X$ defined by
    \begin{equation*}
        \overline{\varepsilon}(p) = p(\varepsilon(p))
    \end{equation*} 
    so that by definition the selection function $\varepsilon$ attains the quantifier $\overline{\varepsilon}$.
\end{enumerate}
\end{definition}

\newcommand{\argmin}{\operatorname{argmin}}
\newcommand{\argmax}{\operatorname{argmax}}

\begin{example}[$\argmin$ and $\argmax$] Let $X$ be a finite non-empty type and $R$ a linearly ordered type. Standard examples of selection functions are \[ \argmin, \argmax \colon (X \to R) \to X, \] which select an arbitrary point at which the outcome function $q \colon X \to R$ attains its minimum or maximum. We will use them in our Tic-Tac-Toe running example. It is easy to see that $\argmin$ and $\argmax$ attain the $\min$ and $\max$ quantifiers, respectively. 
\end{example}

\begin{definition}[Selection function tree] 
Given a type tree $\Xt$, we define the type of \emph{selection function trees over $\Xt$} by
\[
\JJ \Xt = \structure(\J, \Xt).
\]
We adopt the following conventions:
\begin{enumerate}
    \item The variable $\varepsilon$ ranges over the type $\J X$ of selection functions.
    \item The variable $\epsilont$ ranges over the type $\JJ \Xt$ of selection function trees.
    \item The variable $\epsilonf$ ranges over selection function forests $(x \colon X) \to \JJ(\Xf(x))$.
\end{enumerate}
\end{definition}

\begin{example}[Tic-Tac-Toe selection function tree] \label{ttt-s-tree} We associate a selection function tree to the Tic-Tac-Toe type tree of Example \ref{ttt-tree} by alternating the $\argmin$ and $\argmax$ selection functions depending on the level of the tree. Assuming that player X moves first, the root of the selection function tree is the $\argmin$ selection function.
\end{example}

\begin{definition}[Attainability on trees] Given a tree~$\Xt$, a quantifier tree $\phit : \KK \Xt$ and a selection function tree $\epsilont : \JJ \Xt$, we define the relation \begin{quote}
\emph{$\epsilont$ attains $\phit$}     
\end{quote}
by induction on $\Xt$ as follows:
\begin{enumerate}
    \item If $\Xt$ is the empty tree $\emptytree$ then $\epsilont$ attains $\phit$.
    \item If the tree $\Xt$ is of the form $X \cons \Xf$ then $\varepsilon \cons \epsilonf$ attains $\phi \cons \phif$ if $\varepsilon$ attains $\phi$ and for each $x : X$ we have that $\epsilonf(x)$ attains $\phif(x)$. 
\end{enumerate}
\end{definition}

\begin{example} The selection function tree for Tic-Tac-Toe given in Example \ref{ttt-s-tree} has been defined so that it attains the Tic-Tac-Toe quantifier tree of Example \ref{ttt-q-tree}.
\end{example}

\begin{definition}[From selection trees to quantifier trees] Given any tree~$\Xt$ and selection function tree $\epsilont : \JJ \Xt$, we define a quantifier tree \[ \overline{\epsilont} : \KK \Xt \] by induction on $\Xt$ as follows:
\begin{enumerate}
    \item If $\Xt$ is the empty tree $\emptytree$ then $\epsilont$ is the empty tree and so is $\overline{\epsilont}$.
    \item If the tree $\Xt$ is of the form $X \cons \Xf$ then the selection function tree $\epsilont$ is of the form $\varepsilon \cons \epsilonf$ and $\overline{\epsilont}$ is the quantifier tree $\overline{\varepsilon} \cons (\lambda x . \overline{\epsilonf(x)})$. 
\end{enumerate}
A simple proof by induction shows that $\epsilont$ attains $\phit$ if and only if  $\overline{\epsilont} = \phit$. 
\end{definition}

\begin{example} It is easy to see that if $\epsilont$ is the argmin-argmax selection function tree for Tic-Tac-Toe given in Example \ref{ttt-s-tree} then $\overline{\epsilont}$ is the min-max quantifier tree of Example \ref{ttt-q-tree}.
\end{example}

\begin{definition}[Dependent product of selection functions \cite{EO(2010)}]
\label{def-sel-fct-prod}
For any selection function \[ \varepsilon : \J X \]on a type $X$, any a family $Y$ of types indexed by $X$, and any family \[\delta \colon (x : X) \to \J(Y(x))\] of selection functions, we define a selection function \[ \varepsilon \Jtimes \delta : \J((x:X) \times Y(x))\] by
\[
(\varepsilon \Jtimes \delta)(q) = (x_0, \nu(x_0)) 
\]
where
\begin{eqnarray*}
    \nu(x) & = & \delta(x)(q_x) \\
    x_0 & = & \varepsilon(\lambda x . q_x(\nu(x)).
\end{eqnarray*}
\end{definition}

We can also adapt the iterated product of selection functions \cite{EO(2010)} to selection function trees as follows.

\begin{definition}[J-sequence] \label{def-J-sequence}
To each selection function tree $\epsilont : \JJ \Xt$ we associate a single selection function \[ \Jsequence(\epsilont) : \J (\Path \Xt) \] on the type of paths of the tree $\Xt$ by induction as follows.
\begin{enumerate}
    \item To the empty tree $\emptytree : \Tree$ we associate the selection function $\lambda q . \emptypath$ of type $\J (\Path \, \emptytree)$.
    \item To a selection function tree of the form $\varepsilon \cons \epsilonf$ we associate the selection function $\varepsilon \Jtimes (\lambda x . \Jsequence \epsilonf(x))$.
\end{enumerate}
\end{definition}

Note that, by Definition \ref{def-sel-fct-prod}, we have
\[ (\varepsilon \Jtimes (\lambda x . \Jsequence \epsilonf(x)))(q) = x_0 \cons \nu(x_0)
\]
where
\begin{eqnarray*}
    \nu(x) & = & (\Jsequence \epsilonf(x))(q_x) \\
    x_0 & = & \varepsilon(\lambda x . q_x(\nu(x)).
\end{eqnarray*}

\begin{remark}
    Notice the similarity of Definitions~\ref{def-K-sequence} and~\ref{def-J-sequence}. They can be unified in a single definition that assumes that the structure is the functor of a strong monad, which $\K$ and $\J$ are. Moreover, this is also the case for Definitions~\ref{def-quant-prod} and~\ref{def-sel-fct-prod}. Although they look very different, they are both instances of a general definition for strong monads. For any strong monad $T$ we can define
    \[
    t \otimes^T f = (\lambda x. T(\lambda y. (x , y))^\dagger (f (x)))(t)
    \]
    for any $t : T X$ and $f : (x : X) \to T(Y x)$, where $(-)^\dagger$ is the Kleisli extension operator. 
    The definition of the strong monads $K$ and $J$ can be found in~\cite{EO(2010)}, \cite{EO(2010)MR2678125} or~\cite{EO(2010D)} (and also in our Agda formalisation discussed in Section~\ref{sec:formalisation}).
\end{remark}

\subsection{Calculating optimal strategies}

Finally, let us now show how the $\J$-sequence construction can be used to calculate optimal strategies in games based on dependent type trees.

\begin{definition}[Strategy of selection tree] Given tree $\Xt$, a selection function tree $\epsilont : \JJ \Xt$ and an outcome function $q \colon \Path \Xt \to R$, we define a strategy \[\strategy(\epsilont, q) : \Strategy \Xt\] by induction as follows:
\begin{enumerate}
    \item When $\Xt$ is the empty tree then $\strategy(\epsilont, q)$ is the empty strategy.
    \item When $\Xt$ is of the form $X \cons \Xf$ and $\epsilont = \varepsilon \cons \epsilonf$ then we define \[\strategy(\epsilont, q) = (x_0, \sigmaf)\] where
\begin{eqnarray*}
    x_0 & = & \head(\Jsequence(\epsilont)(q)), \\
    \sigmaf(x) & = & \strategy(\epsilonf(x), q_x).
\end{eqnarray*}
\end{enumerate}
\end{definition}

\begin{lemma}[Main lemma] \label{main-lemma} Given a tree $\Xt$, a selection function tree $\epsilont : \JJ \Xt$ and an outcome function $q \colon \Path \Xt \to R$, we have
\begin{equation*}
    \spath(\strategy(\epsilont, q)) = \Jsequence(\epsilont)(q). 
\end{equation*}
\end{lemma}

\begin{proof} By induction on $\Xt$. When $\Xt$ is the empty tree, the result is immediate. Let $\Xt = X \cons \Xf$ and $\epsilont = \varepsilon \cons \epsilonf$ and assume, for each $x : X$, the induction hypothesis holds
\begin{equation*}
    \spath(\strategy(\epsilonf(x), q_x)) = \Jsequence(\epsilonf(x))(q_x) 
\end{equation*}
Then
\[
\begin{array}{rcl}
    \spath(\strategy(\epsilont, q)) 
        & = & \spath(x_0 \cons \sigmaf) \\[2mm]
        & = & x_0 \cons \spath(\sigmaf(x_0)) \\[1mm]
        & \stackrel{(\text{IH})}{=} & x_0 \cons \nu(x_0) \\[2mm]
        & = & \Jsequence(\epsilont)(q)
\end{array}
\]
where
\begin{eqnarray*}
    \nu(x) & = & \Jsequence (\epsilonf(x))(q_x) \\
    x_0 & = & \varepsilon(\lambda x . q_x(\nu(x))) \\
    \sigmaf(x) & = & \strategy(\epsilonf(x), q_x),
\end{eqnarray*}
which concludes the proof. \hfill $\Box$
\end{proof}

\begin{lemma}[Selection strategy lemma] \label{sel-strat-lemma} Given a type tree $\Xt$, an outcome function $q \colon \Path \Xt \to R$, and a selection function tree $\epsilont : \JJ \Xt$, \[ \strategy(\epsilont, q) : \Strategy \Xt\] is optimal for the game $(\Xt, R, q, \overline{\epsilont})$.
\end{lemma}

\begin{proof} By induction on the tree $\Xt$. When $\Xt$ is the empty tree, the result is again immediate. If $\Xt = X \cons \Xf$ and $\epsilont = \varepsilon \cons \epsilonf$, recall that $\strategy(\epsilont, q) = x_0 \cons \sigmaf$ where
\begin{eqnarray*}
    x_0 & = & \varepsilon(\lambda x . q_x(\Jsequence (\epsilonf(x))(q_x))), \\
    \sigmaf(x) & = & \strategy(\epsilonf(x), q_x).
\end{eqnarray*}
By induction hypothesis we have that, for all $x : X$, the strategy $\sigmaf(x)$ is optimal for the subgame \[ (\Xf(x), R, q_x, \overline{\epsilonf(x)}) \] and since \[ x_0 = \varepsilon(\lambda x . q_x(\spath(\strategy(\epsilonf(x), q_x)))), \] by the Main lemma \ref{main-lemma}, we indeed have
\[ 
q(x_0 \cons \spath (\sigmaf(x_0))) = \overline{\varepsilon}(\lambda x . q(x \cons \spath (\sigmaf(x)))),
\]
which concludes the proof. \hfill $\Box$
\end{proof}

\begin{theorem}[Selection strategy theorem] If a selection function tree $\epsilont \colon \JJ \Xt$ attains a quantifier-free $\phit \colon \KK \Xt$ then, for any outcome function $q \colon \Path \Xt \to R$, \[ \strategy(\epsilont, q) \] is optimal for the game $(\Xt, R, q, \phit)$. 
\end{theorem}

\begin{proof} Since $\overline{\epsilont} = \phit$ when  $\epsilont$ attains $\phit$, the theorem is an immediate consequence of the Selection strategy lemma \ref{sel-strat-lemma}.
\end{proof}

\begin{example}[Tic-Tac-Toe optimal strategy] Since the quantifier tree $\phit$ in the Tic-Tac-Toe game $(\Xt, R, q, \phit)$ (Example \ref{def:ttt-game}) is attained by a selection function tree $\epsilont$ (Definition \ref{ttt-s-tree}), the theorem above shows that $\strategy(\epsilont, q)$ is optimal.
\end{example}





\section{Discussion of dependent type trees}
\label{sec:discussion}

\newcommand{\A}{\Tree_A}
\newcommand{\isempty}[1]{{\rm is{-}empty}(#1)}
\newcommand{\prune}{\operatorname{prune}}

Aczel \cite{Aczel(1978)} considers a type $\A$ similar to our type $\Tree$, but without our leaf constructor $[]$, for the purpose of modelling CZF. The type $\A$ consists of well-founded type trees where leaves are of the form $X \cons \Xf$ with~$X$ empty and~$\Xf$ the unique function on~$X$. 
In our definition, nodes~$X$ with~$X$ empty are allowed in principle. 
However, as soon as there is a strategy for a type tree $\Xt : \Tree$, the internal nodes of~$\Xt$ must all be inhabited. 

It is not hard to see that the subtype of $\Tree$ consisting of the trees whose internal nodes are all inhabited is in bijection with the subtype of $\A$ consisting of trees whose internal nodes are all decidable (that is, either empty or inhabited). So we can work with the type $\A$ instead. But, in this case, we need to slightly modify all our constructions on trees. 
For example, with our definition of path, the type of paths over a tree~$\Xt : \Tree$ is empty if $\Xt$ has no leaves of the form $[]$. However, we can define a correct notion of path for Aczel trees by
\[
\Path (X \cons \Xf) = \isempty{X} + ((x : X) \times \Path (\Xf(x)). 
\]
This definition is correct in that paths in our original sense correspond to paths in this modified sense along the bijection discussed above. All other definitions have to be modified in a similar way, by treating the base case differently.

Another observation is that empty internal nodes can be pruned away with an idempotent map $\prune : \Tree \to \Tree$ such that any tree $\Xt$ has the same paths as the tree $\prune(\Xt)$, and this can be done without assuming that the internal nodes of $\Xt$ are decidable. Moreover, as soon as there is at least one path in the tree $\Xt$, all internal nodes of tree $\prune(\Xt)$ are inhabited.

An extended version of this discussion, with full technical details, is formalised in Agda (see the next section).

\section{Proofs and Programs in Agda}
\label{sec:formalisation}

A second contribution of the present work is that all the definitions, constructions, theorems, lemmas, proofs and examples presented above have been implemented and computer-verified in Agda \cite{TypeTopologyGames}. 

As Agda adopts a constructive foundation, the function $\strategy(\epsilont, q)$, which defines strategies from selection trees~$\epsilont$ and outcome functions~$q$, is automatically a computer program that can be used to calculate optimal strategies in concrete games. In particular, we also implemented Tic-Tac-Toe as outlined in the above examples.

The Agda development discussed above \cite{TypeTopologyGames} is organised in the following files:

\begin{enumerate}
    \item \emph{TypeTrees.} This defines dependent type trees and basic concepts for them, including the general definition of adding ``structure'' to a type tree.
    \item \emph{FiniteHistoryDependent}. This implements this paper, excluding the examples. We took care of presenting the material in the same order as in this paper, and of using the same labels for theorems, definitions, etc.
    \item There are three \emph{TicTacToe} files, which implement our running example in different ways. The first one is the shortest and most transparent, but also the most inefficient, while the other two try to sacrifice clarity for efficiency. 
    \item \emph{Constructor}. This gives an alternative construction of games, not discussed in this paper, and is used in the more efficient versions of Tic-Tac-Toe.
    \item \emph{Discussion}. This formalises an extended version of the discussion of Section~\ref{sec:discussion}.
\end{enumerate}



\section{Further work}

The work presented here is part of a project to extend the theory of higher-order games so as to model different types of players and strategies. Our aim is to make use of the combination of the selection monad with other monads, as in \cite{EO(2017)}, leading to games and strategies parametrised by a monadic structure. In particular, by making use of a probability monad, we intend to model suboptimal or ``irrational'' players in future work.

\section*{Acknowledgements}

We thank the anonymous referee for a number of helpful comments and questions.
We would also like to thank Ohad Kammar for discussions. He also implemented~\cite{KammarGames} part of the above in the dependently typed programming language Idris~\cite{JFP:9060502}.









 \bibliographystyle{plain} 
 \bibliography{references}

\end{document}